\documentclass[12pt]{article}

\usepackage{amsmath, amsthm, amssymb, amsfonts}
\DeclareMathOperator{\rank}{rank}
\newcommand{\ap}[2]{#1{:}#2}
\DeclareMathOperator{\re}{RE}
\DeclareMathOperator{\da}{DA}
\DeclareMathOperator{\eada}{EADA}

\usepackage[margin=1.25in]{geometry}
\usepackage{setspace}
\usepackage[dvipsnames]{xcolor}
\usepackage{tikz}
\usepackage{subfiles}
\usepackage{ifthen}

\theoremstyle{plain}
\newtheorem{theorem}{Theorem}
\newtheorem{lemma}{Lemma}
\newtheorem{proposition}{Proposition}
\newtheorem{corollary}{Corollary}
\theoremstyle{definition}

\newtheorem{definition}{Definition}

\usepackage{natbib}
\usepackage[
    colorlinks=true,
    linkcolor=blue,
    citecolor=blue,
    urlcolor=blue,
    filecolor=blue,
    hyperfootnotes=false,
    hyperindex,
    breaklinks
]{hyperref}

\title{Rawlsian equity: a new notion of fairness for the assignment problem}

\author{
\"Ozg\"un Ekici\thanks{Department of Economics, \"Ozye\u{g}in University, Istanbul, T\"urkiye. Email: \texttt{ozgun.ekici@ozyegin.edu.tr}}
\and
Sinan Ertemel\thanks{Department of Economics, Istanbul Technical University, Istanbul, T\"urkiye. Email: \texttt{ertemels@itu.edu.tr}}
\and
M. Bumin Yenmez\thanks{Department of Economics, Washington University in St. Louis, USA; Durham University, UK; and \"Ozye\u{g}in University, T\"urkiye. Email: \texttt{bumin@wustl.edu}}
}

\date{\today}

\begin{document}
\maketitle

\begin{abstract}
We introduce Rawlsian equity, a notion of fairness for allocating indivisible objects among agents without object-specific entitlements. Rawlsian equity requires that an object not be assigned to an agent who ranks it more highly than another agent unless the latter receives a more preferred object. We show that the set of Rawlsian equitable allocations coincides with the set of stable allocations in an auxiliary market where object priorities depend on agents' preference reports. The agent-proposing deferred acceptance algorithm computes the agent-optimal element of this set, but no Rawlsian equitable rule is strategy-proof in general. We also study Rawlsian efficiency, which reconciles our fairness notion with Pareto efficiency; and we characterize Rawlsian equity among fairness criteria that adjudicate objections pairwise on the basis of ranks.
\end{abstract}



\section{Introduction}
If beauty is in the eye of the beholder, then so is fairness. Consider the canonical story of a flute that must be given to one of four children. The first repaired the flute and claims it as a reward for her effort. The second is the daughter of its previous owner and claims it as her birthright. The third is a flutist and claims it because she can put it to the best use. The fourth has far fewer toys than the others and claims it on compensatory grounds. These four claims appeal, respectively, to the principles of reward, exogenous rights, fitness, and compensation.\footnote{See \citet{moulin:04} for the four principles of fairness; our version of the flute story is from p.~21, where it is traced to Plato. See also Aristotle, \textit{Politics}, III.12, 1282b--1283a.}

The first three principles have played a central role in matching and market design. The reward principle, for example, underlies office assignments based on seniority. The exogenous rights principle underlies kidney exchange, where patients are treated as having property rights over their donors' kidneys. The fitness principle underlies college admissions, where applicants with higher test scores receive priority because test scores are regarded as indicators of how well applicants will perform academically in college.

In contrast, the compensation principle has received less formal attention in this literature, although in many applications it is the only one with any normative purchase. When agents are indistinguishable in the eyes of the allocating authority, reward, exogenous rights, and fitness are inapplicable. For instance, in the allocation of public housing units to families or dormitory rooms to incoming students, agents have neither contributed to the provision of the allocated resources nor possess prior claims over them. Therefore, the reward and exogenous rights principles are inapplicable. Arguably, in these applications, fitness is not a particularly relevant basis for preferential treatment either. This leaves the compensation principle as the only potentially applicable one among the four.

Motivated by the compensation principle, we reconsider the assignment problem in light of a new notion of fairness. In this classical problem, there are $n$ indivisible objects to be allocated to $n$ agents: each agent receives exactly one object, preferences are strict, and monetary transfers are not allowed. We ask what compensation requires when the only information that distinguishes agents is their preference rankings over objects. As a guiding principle, compensation seeks to improve the welfare of agents who are worse off. We deploy this objective by developing an ordinal version of the Rawlsian principle, which accords greater weight to the welfare of less advantaged agents.

Suppose agent $i$ ranks object $o$ as her fourth choice, and agent $j$ ranks it as her second choice. If $j$ is assigned $o$, $i$ will not envy her provided that she is assigned one of her top three choices. However, if she is assigned an object worse than $o$, she envies $j$ for having $o$. She may object to $j$'s assignment of $o$ as her second choice when she cannot receive it even as her fourth choice. We describe this situation as Rawlsian envy. Formally, agent $i$ \emph{Rawlsian envies} agent $j$ if $i$ prefers $j$'s assigned object to her own, and this object ranks lower in $i$'s preference order than in $j$'s. If both agents rank $j$'s assigned object in the same position, the tie is broken by means of an exogenous tie-breaking profile: $i$ then Rawlsian envies $j$ if the tie-breaking order for that object ranks $i$ ahead of $j$. We call an allocation \emph{Rawlsian equitable} if no agent Rawlsian envies another.

The reference to Rawls in the name of our equity notion is intended in a narrow sense. The Rawlsian principle of giving priority to the less advantaged \citep{rawls:58} can also be applied using cardinal measures of welfare. However, in the assignment problem, we work with ordinal preferences; therefore, we apply the Rawlsian principle in the ordinal domain. Even in the ordinal domain, ours is not the only way to apply the Rawlsian principle. Another interpretation is the leximin criterion, under which the goal is to make the worst assigned rank as favorable as possible, then the second-worst, and so on. Arguably, the leximin criterion reflects the Rawlsian principle more directly. However, under the leximin criterion, for an agent to determine whether her assignment is justified, she must compare entire allocations, which may not be easy when $n$ is large. In contrast, under Rawlsian equity, an agent need only compare her assignment with that of another agent to determine whether she Rawlsian envies that agent and, consequently, whether her assignment is justified. In this sense, Rawlsian equity is a local fairness criterion because it compares assignments pairwise, whereas the leximin criterion is global because it compares entire allocations.

The exogenous tie-breaking profile turns out to play a smaller role than one might expect. An allocation can be made Rawlsian equitable by an appropriate choice of the tie-breaking profile if and only if no agent ranks a contested object strictly worse than its holder does, and it is Rawlsian equitable under every choice of the tie-breaking profile if and only if, for every contested object, no agent has even a tied claim against its holder (\autoref{prop:tiebreak}). The tie-breaking profile thus has complete authority over equal claims and none at all over unequal ones.

Our notion of Rawlsian equity also has a two-sided matching interpretation: Given agents' preferences, construct an auxiliary matching market where each object has its own separate priority order. An agent receives a higher priority for an object if she ranks it lower in her preference ordering, with ties broken by the object's tie-breaking order. Agent $i$ then Rawlsian envies agent $j$ if and only if agent $i$ and the assigned object of agent $j$ form a blocking pair in this auxiliary market. Hence, the Rawlsian equitable allocations are exactly the stable allocations of this auxiliary market. Consequently, the set of Rawlsian equitable allocations is non-empty, forms a lattice, and contains an agent-optimal allocation, namely the one produced by the agent-proposing deferred acceptance algorithm (Theorems~\ref{thm:existence} and~\ref{thm:agent-optimal}). Thus, the study of the assignment problem under Rawlsian equity closely resembles the study of the school choice model under stability.

However, the resemblance to the school-choice model does not extend to incentives. In school choice, priorities are exogenous, and the deferred acceptance rule is strategy-proof for the proposing side \citep{dubins/freedman:81,roth:82}. In contrast, in our model, priorities in the auxiliary market are determined by agents' reported preferences (together with the tie-breaking profile). This gives an agent the power to misreport her preferences and thereby reshape priorities to her advantage. Consequently, the agent-proposing deferred acceptance rule is not strategy-proof (\autoref{cor:da-not-sp}). Indeed, we show that Rawlsian equity is incompatible with strategy-proofness (\autoref{prop:impossible}).

Rawlsian equity also conflicts with efficiency. We provide an example in which the agent-optimal Rawlsian equitable allocation is Pareto dominated and, consequently, no allocation is both Rawlsian equitable and Pareto efficient (\autoref{prop:da-notefficient} and \autoref{cor:no-eff-equitable}). We then ask what survives of the equity requirement once efficiency is imposed. We show that the full-consent efficiency-adjusted deferred acceptance rule of \citet{kesten:10a}, applied to the auxiliary market, selects the unique Rawlsian efficient allocation, which either coincides with the agent-optimal Rawlsian equitable allocation or Pareto dominates every Rawlsian equitable allocation (\autoref{thm:eada}).

Finally, we ask what distinguishes Rawlsian equity among fairness notions of its kind. Consider notions that adjudicate objections pairwise, based on the positions that the two agents assign to the contested object and nothing else. Four conditions single out Rawlsian equity within this class: every comparison of competing claims has exactly one winner, the agent who ranks the contested object lower prevails, the resolution of a tie does not depend on the position at which the tie occurs, and a fair allocation exists at every preference profile (\autoref{thm:rawlsian-characterization}). The tie-breaking profile is not assumed here but derived, because resolving equal claims cyclically would leave some preference profiles with no fair allocation.

Our paper connects to several lines of research. The closest antecedent is \citet{bogomolnaia/moulin:01}, who study fairness in the assignment problem using only ordinal information, as we do. Because no deterministic allocation of indivisible objects is generally envy-free, they restore fairness ex ante through randomization: the probabilistic serial rule is ordinally efficient and envy-free in the sense of first-order stochastic dominance. We take the complementary approach. Rather than enlarging the set of outcomes, we weaken the fairness requirement itself, grading agents' claims to an object by how they rank it, so that a fair deterministic allocation exists for every preference profile (\autoref{thm:existence}).

The closest criterion to ours is the rank envy-freeness of \citet{belahcene/mousseau/wilczynski:21}, who study the same environment and, like us, compare the positions two agents assign to a contested object. The two criteria read the comparison in opposite directions. Under rank envy-freeness, agent $i$'s objection to $j$ holding $o$ is upheld when $i$ ranks $o$ better than $j$ does; under Rawlsian equity, it is upheld when $i$ ranks $o$ worse, with equal ranks settled by the tie-breaking order. The difference is akin to the difference between fitness and compensation: the first standard awards the object to the agent more eager for it; the second awards it to the agent for whom it is the less favorable fallback.

There are also prior studies that invoke the Rawlsian principle. \citet{afacan/dur:24} study deterministic object allocation and call an allocation Rawlsian when neither the rank of the worst-off agent nor the number of agents assigned that rank can be improved. \citet{demeulemeester/pereyra:24} work in the random assignment model and select the assignment that lexicographically improves the position of the worst-off agent. Both criteria rank complete allocations, and both are incompatible with strategy-proofness, as is Rawlsian equity. Our notion differs in being objection-based, as it compares the assignments of pairs of agents rather than entire allocations. Relatedly, mechanisms are sometimes compared through the rank distributions they induce, for example, by the average assigned rank or by the rank of the worst-off agent \citep{ortega/klein:23}.

Our paper also relates to justified envy in priority-based matching \citep{balinski/sonmez:99,abdulkadiroglu/sonmez:03}. In that literature, priorities are exogenously given, whereas in our auxiliary market, reported preferences shape priorities. \citet{troyan/delacretaz/kloosterman:20} weaken stability by discounting claims whose enforcement sets off a chain of reassignments that costs the claimant the object she claimed, and connect this to \citet{kesten:10a}. Their notion discounts objections when honoring them would be self-defeating; ours discounts them when the claimant is not the worse-positioned party. Our efficiency analysis draws on \citet{tang/yu:14,reny:22,ehlers/morrill:20}.

\section{The model}\label{sec:model}

\subsection{Preliminaries}
An \textbf{assignment problem} consists of a set of agents $N=\{1,2,\allowbreak \ldots,n\}$ and a set of objects $O$. We assume that each agent $i\in N$ has a strict preference relation $P_i$ over $O$, which is a strict linear order. Throughout, we also assume that $|N| = |O|$; we discuss relaxing this assumption at the end of \autoref{sec:eada}.

We denote the profile of agent preferences by $P=(P_i)_{i\in N}$ and the set of strict preference relations over $O$ by $\mathcal{P}$, so $P_i\in\mathcal{P}$ and $P\in\mathcal{P}^n$. We write $o \mathrel{P_i} o'$ if agent $i$ prefers object $o$ to $o'$. We also present $P_i$ as an ordering, as in $o\succ o'\succ o''$, listing objects from the most-preferred to the least-preferred. We define the associated weak preference relation $R_i$ by $o \mathrel{R_i} o'$ if $o \mathrel{P_i} o'$ or $o=o'$.

An \textbf{allocation} $\mu:N\rightarrow O$ is a one-to-one, and hence onto, mapping from the set of agents to the set of objects. For $i\in N$, $\mu(i)\in O$ denotes agent $i$'s assignment under $\mu$. Likewise, for $o\in O$, $\mu^{-1}(o)\in N$ denotes the agent assigned to $o$. In examples, we label objects as $a$, $b$, $c$, and so forth, and we denote an allocation $\mu$ by listing each agent with her assignment, as in $(\ap{1}{a},\ap{2}{b},\ap{3}{c})$, which means that $\mu(1)=a$, $\mu(2)=b$, and $\mu(3)=c$. We denote the set of all allocations by $\mathcal{M}$.

A \textbf{rule} $\phi:\mathcal{P}^n\rightarrow\mathcal{M}$ is a mapping from the set of preference profiles to the set of allocations. It associates each profile $P\in \mathcal{P}^n$ with an allocation $\phi(P)\in \mathcal{M}$. For $i\in N$, $\phi_i(P)\in O$ denotes agent $i$'s assignment under $\phi(P)$.

An allocation $\mu\in \mathcal{M}$ is \textbf{Pareto efficient} at $P\in \mathcal{P}^n$ if there exists no allocation $\mu'\in \mathcal{M}$ such that $\mu'(i)\mathrel{R_i}\mu(i)$ for each $i \in N$, with strict preference for some agent.

Henceforth, we say that a rule $\phi$ satisfies a property of allocations if, for every $P\in\mathcal{P}^n$, the allocation $\phi(P)$ satisfies that property at $P$. For instance, $\phi$ is Pareto efficient if, for every $P\in\mathcal{P}^n$, $\phi(P)$ is Pareto efficient at $P$.

A rule $\phi$ is \textbf{strategy-proof} if no agent can obtain a preferred assignment by misreporting her preferences: for every $P\in\mathcal{P}^n$, every agent $i\in N$, and every $P'_i\in\mathcal{P}$,
\[ \phi_i(P)\mathrel{R_i}\phi_i(P'_i,P_{-i}), \]
where $(P'_i,P_{-i})$ denotes the profile obtained from $P$ when agent $i$'s preference relation $P_i$ is replaced by $P'_i$ and $P_{-i}=(P_j)_{j\neq i}$.

\subsection{Rawlsian equity}\label{sec:equity}

Given a preference relation $P_i$, let $\rank(P_i,o) = |\{o'\in O : o' \mathrel{R_i} o\}|$ denote the position of object $o$ at $P_i$. For instance, the most-preferred object at $P_i$ has rank $1$, the second most-preferred has rank $2$, and so on.

Rawlsian equity, our fairness notion, is inspired by the Rawlsian idea of according greater weight to the welfare of less advantaged agents: Suppose object $o$ is agent $i$'s fourth and agent $j$'s second choice. Rawlsian equity holds that placing agent $i$ at her fourth choice takes precedence over placing agent $j$ at her second, so $o$ should not be assigned to $j$ unless agent $i$ receives an object she prefers over $o$. Put differently, Rawlsian equity regards an agent who ranks an object less favorably as having a stronger claim to it and requires that stronger claims be respected. We formalize this idea next.

For an object $o\in O$ and distinct agents $i,j\in N$, agent $i$ has a \textbf{stronger claim} to $o$ than agent $j$ if $\rank(P_i,o) > \rank(P_j,o)$, and the two agents have \textbf{equal claims} to $o$ if $\rank(P_i,o) = \rank(P_j,o)$.

The strong version of Rawlsian equity rules out envy even under equal claims.

\begin{definition}[Absolute Rawlsian equity]
At a profile $P\in\mathcal{P}^n$ and an allocation $\mu\in\mathcal{M}$, agent $i\in N$ \textbf{weakly Rawlsian envies} agent $j\in N$ if $\mu(j) \mathrel{P_i} \mu(i)$ and agent $i$'s claim to $\mu(j)$ is stronger than or equal to agent $j$'s. An allocation $\mu\in\mathcal{M}$ is \textbf{absolutely Rawlsian equitable} at $P\in\mathcal{P}^n$ if no agent weakly Rawlsian envies another.
\end{definition}

However, this demanding fairness condition is not always attainable: for a given preference profile, an absolutely Rawlsian equitable allocation may not exist. To see this, suppose there are two agents and two objects, $a$ and $b$, with $a \mathrel{P_1} b$ and $a \mathrel{P_2} b$. Then, under either allocation, the agent who receives $b$ weakly Rawlsian envies the agent who receives $a$: she prefers $a$ to $b$, and the two agents have equal claims to $a$. It turns out that the impossibility stems precisely from these unresolved ties.

In applications, equally deserving agents routinely have competing claims to the same object. Institutions resolve such conflicts through an exogenous tie-breaking order. The tie-breaker may be the waiting time when allocating public housing units to families, or it may be a lottery number when allocating school seats to students. Our next notion of Rawlsian equity reflects this institutional practice: when agents have equal claims, ties are resolved according to an exogenously given tie-breaking profile.

A \textbf{tie-breaking profile} $\pi=(\pi_o)_{o \in O}$ consists of a \textbf{tie-breaking order} $\pi_o$ of the agents for each object $o$, where each tie-breaking order is a strict linear order on $N$. We write $\pi_o(i)=k$ if agent $i\in N$ is ordered $k^{\text{th}}$ under $\pi_o$. The set of strict linear orders on $N$ is denoted by $\Pi$, so $\pi_o\in\Pi$ and $\pi\in\Pi^{n}$. As with preferences, we also present a strict linear order on $N$ as an ordering, as in $i\succ j\succ k$.

The next version of Rawlsian equity we introduce is weaker because it permits envy under equal claims if the tie-breaking order ranks the envious agent below the agent holding the object.

\begin{definition}[Rawlsian equity]
At a profile $P\in\mathcal{P}^n$, an allocation $\mu\in\mathcal{M}$, and a tie-breaking profile $\pi\in\Pi^n$, agent $i\in N$ \textbf{Rawlsian envies} agent $j\in N$ at $\pi$ if $\mu(j) \mathrel{P_i} \mu(i)$, and either agent $i$'s claim to $\mu(j)$ is stronger than agent $j$'s, or the two agents have equal claims to $\mu(j)$ and $\pi_{\mu(j)}(i) < \pi_{\mu(j)}(j)$. An allocation $\mu\in\mathcal{M}$ is \textbf{Rawlsian equitable} at $\pi$ if no agent Rawlsian envies another at $\pi$.
\end{definition}

Rawlsian envy is the form of envy that we regard as justified. Under allocation $\mu$, agent $i$ envies agent $j$ when she prefers $j$'s object, $\mu(j)$, to her own object, $\mu(i)$. She may complain that $\mu(j)$ has not been assigned to her, but whether this complaint is justified depends on the two agents' claims to the contested object, $\mu(j)$. If agent $i$'s claim is stronger, the complaint is justified because she is the more disadvantaged agent. If the two agents have equal claims, the tie-breaking profile settles the matter: the complaint is justified if the tie-breaking order for $\mu(j)$ ranks $i$ ahead of $j$.

For instance, let $N=\{1,2,3\}$ and $O=\{a,b,c\}$. Let $P_1: b\succ a\succ c$, $P_2: a\succ b\succ c$, and $P_3: c\succ a\succ b$. Then, the allocation $(\ap{1}{b}, \ap{2}{a}, \ap{3}{c})$ is Rawlsian equitable at every tie-breaking profile since it assigns each agent her top choice. In contrast, the allocation $(\ap{1}{c}, \ap{2}{a}, \ap{3}{b})$ is Rawlsian equitable at no tie-breaking profile: agent $1$ prefers $a$ to her assigned object $c$, and her claim to $a$ is stronger than agent $2$'s, so she Rawlsian envies agent $2$ regardless of the tie-breaking profile.

Justified envy always reflects a genuine disadvantage. Indeed, combining $\mu(j) \mathrel{P_i} \mu(i)$ with the comparison of claims yields
\[ \rank(P_i,\mu(i)) > \rank(P_i,\mu(j)) \geq \rank(P_j,\mu(j)), \]
so an agent who Rawlsian envies another receives an object that she ranks strictly worse than the other agent ranks her own object.

Two features of the definition deserve emphasis. First, Rawlsian envy compares ranks across agents: whether agent $i$'s complaint against agent $j$ is justified turns on how the two agents rank the contested object. Treating positions in different preference orders as comparable is a substantive assumption, which we adopt because ranks are the only interpersonally comparable information that ordinal reports contain about how much an agent values an object. Second, the comparison concerns the contested object alone, not the agents' overall positions. Suppose agent $i$ receives her fifth choice and prefers agent $j$'s object, which she ranks first and agent $j$ ranks second. Although agent $i$ is worse off by assigned rank, she has no Rawlsian claim: with respect to the contested object, she is the agent the fitness principle would favor, whereas compensation sides with the agent for whom the object is the less favorable fallback. Objections grounded in overall positions belong to global criteria such as leximin; Rawlsian equity confines the comparison to the object in dispute, which is what makes it a local criterion.

Weak Rawlsian envy and Rawlsian envy differ only in their treatment of equal claims. The first counts envy under equal claims as justified; the second counts it only if the tie-breaking order ranks the envier ahead of the holder of the envied object. \autoref{prop:tiebreak} shows how the relationship between these two fairness notions depends on the tie-breaking profile.

\begin{proposition}[Role of the tie-breaking profile]\label{prop:tiebreak}
Let $P\in\mathcal{P}^n$ and $\mu\in\mathcal{M}$. Then $\mu$ is Rawlsian equitable at every tie-breaking profile if and only if it is absolutely Rawlsian equitable, and $\mu$ is Rawlsian equitable at some tie-breaking profile if and only if no agent prefers another agent's object to her own while having a stronger claim to it.
\end{proposition}

All omitted proofs are in the Appendix.

The proposition parallels two standard results concerning stable matchings in the presence of indifferences. Order the agents at each object by their claims, leaving equal claims tied: object priorities become weak orders, agent preferences remain strict, and Rawlsian equity at a tie-breaking profile is stability under the corresponding refinement of the tied priorities. It then restates, for this market, the facts that a matching is weakly stable if and only if it is stable under some refinement and super-stable if and only if it is stable under every refinement \citep{irving:94,manlove:02}; in particular, absolute Rawlsian equity is super-stability. The parallel ends there: that theory takes the weak priorities as given, whereas ours arise from the reported preferences themselves, and this endogeneity is responsible for the incentive results below. That literature also offers no analog of our characterization.

In the rest of the paper, we fix the tie-breaking profile $\pi\in\Pi^n$ and suppress it in the notation unless it is explicitly needed.

\section{Results}\label{sec:results}

\subsection{Existence and lattice structure}\label{sec:lattice}
Even though an absolutely Rawlsian equitable allocation need not exist, a Rawlsian equitable allocation always does.

\begin{theorem}[Existence]\label{thm:existence} For each $P\in\mathcal{P}^n$, there exists an allocation that is Rawlsian equitable at $P$. \end{theorem}

This result and its generalization below both rest on an equivalence between Rawlsian equity in the assignment problem and stability in an auxiliary matching market, in the sense of \citet{gale/shapley:62}. The two problems consist of the same agents, objects, and agent preferences, but differ in the information associated with each object: in our setting, each object is endowed with a tie-breaking order that separates agents with equal claims, whereas in the auxiliary matching market, each object has a priority ranking over agents, constructed from the agent preferences and the tie-breaking order.

For each $P\in\mathcal{P}^n$, construct a \textbf{priority profile} $\pi^{P}=(\pi^{P}_o)_{o\in O}$, where each $\pi^{P}_o$ is a priority order over agents for object $o$; we write $i \mathrel{\pi^{P}_o} j$ if agent $i$ has higher priority than agent $j$ for $o$. For each $o\in O$ and distinct $i,j\in N$, set $i \mathrel{\pi^{P}_o} j$ if agent $i$ has a stronger claim to $o$ than agent $j$, or if the two agents have equal claims to $o$ and $\pi_o(i) < \pi_o(j)$. In words, $\pi^{P}_o$ ranks agents by the strength of their claims to $o$, breaking ties according to the tie-breaking order $\pi_o$. Because $\pi^{P}_o$ orders agents lexicographically, first by the strength of their claims and then by $\pi_o$, it is itself a strict linear order. Hence, $\pi^{P}\in\Pi^n$.

The \textbf{agent-proposing deferred acceptance algorithm}, run with the priority profile $\pi^{P}$, proceeds as follows.
\begin{itemize}
\item \emph{Round $1$.} Each agent proposes to her most-preferred object. Each object provisionally holds the highest-priority proposer according to $\pi^{P}_o$ and rejects the rest.
\item \emph{Round $k>1$.} Each agent rejected in round $k-1$ proposes to her most-preferred object among those that have not yet rejected her. Each object provisionally holds, among its currently held agent and the agents newly proposing to it, the one with the highest $\pi^{P}_o$-priority, and rejects the rest.
\end{itemize}

The rounds continue until no agent is rejected, at which point the provisional holdings become final. Since $|N|=|O|$ and every agent ranks all objects, the algorithm always produces an allocation, which we denote by $\da(P)$. The corresponding allocation rule is denoted by $\da$.

We compare allocations by agents' preferences over their assignments: for $\mu,\mu'\in\mathcal{M}$, write $\mu\succeq_N\mu'$ if $\mu(i)\mathrel{R_i}\mu'(i)$ for every $i\in N$. A set $\mathcal{A}\subseteq\mathcal{M}$ forms a \textbf{lattice} under $\succeq_N$ if, for all $\mu,\mu'\in\mathcal{A}$, the mapping that gives each agent the more-preferred of $\mu(i)$ and $\mu'(i)$, and the mapping that gives each agent the less-preferred of $\mu(i)$ and $\mu'(i)$, are themselves allocations in $\mathcal{A}$. (These are the join and the meet of $\mu$ and $\mu'$ under $\succeq_N$.) An allocation is \textbf{agent-optimal} in $\mathcal{A}$ if it belongs to $\mathcal{A}$ and every agent weakly prefers it to every other allocation in $\mathcal{A}$. Since preferences are strict, there is at most one agent-optimal allocation in $\mathcal{A}$.

\begin{theorem}[Lattice structure]\label{thm:agent-optimal} For each $P\in\mathcal{P}^n$, the set of Rawlsian equitable allocations at $P$ is a non-empty lattice under $\succeq_N$, and $\da(P)$ is the agent-optimal Rawlsian equitable allocation. \end{theorem}

Because \autoref{thm:agent-optimal} asserts in particular that the set of Rawlsian equitable allocations is non-empty, it generalizes \autoref{thm:existence}.

\begin{proof}
Fix a priority profile $\rho\in\Pi^n$. An allocation $\mu$ is \textbf{stable} with respect to $\rho$ if there is no pair $(i,o)$ with $o \mathrel{P_i} \mu(i)$ and $i \mathrel{\rho_o} \mu^{-1}(o)$; that is, no agent $i$ prefers some object $o$ to her assignment while having a higher $\rho_o$-priority than the agent who holds $o$. Such a pair is a \textbf{blocking pair}.

We claim that $\mu$ is Rawlsian equitable at $P$ if and only if $\mu$ is stable with respect to $\pi^{P}$. Consider distinct agents $i,j\in N$ and let $o=\mu(j)$. By construction, $i \mathrel{\pi^{P}_o} j$ if and only if $i$ has a stronger claim to $o$ than agent $j$, or the two agents have equal claims to $o$ and $\pi_o(i) < \pi_o(j)$. Given that $\mu(j) \mathrel{P_i} \mu(i)$, this is exactly the condition under which $i$ Rawlsian envies $j$. Hence, agent $i$ Rawlsian envies agent $j$ if and only if $(i,\mu(j))$ is a blocking pair for $\mu$ with respect to $\pi^{P}$. Therefore, $\mu$ has no Rawlsian envy pair if and only if it has no blocking pair.

Since $\pi^{P}\in\Pi^n$, the set of allocations stable with respect to $\pi^{P}$ is non-empty, and agent-proposing deferred acceptance, run with priorities $\pi^{P}$, produces the allocation $\da(P)$, which is agent-optimal within that set \citep{gale/shapley:62}. Moreover, the set of stable allocations forms a lattice under $\succeq_N$ \citep{knuth:76}. By the equivalence, the set of Rawlsian equitable allocations at $P$ coincides with the set of stable allocations. Therefore, the set of Rawlsian equitable allocations also forms a lattice under $\succeq_N$, and $\da(P)$ is the agent-optimal Rawlsian equitable allocation.
\end{proof}

\subsection{Incentives}\label{sec:incentives}

Because Rawlsian equity coincides with stability in this auxiliary market, the structural properties of the set of stable allocations at a fixed profile transfer directly, as illustrated by its lattice structure and the agent-optimality of $\da(P)$. Properties that compare across profiles, however, transfer less readily since changing $P$ alters both the agents' preferences and the constructed priorities $\pi^{P}$. Incentives are the leading example: in a market with fixed priorities, agent-proposing deferred acceptance is strategy-proof \citep{dubins/freedman:81,roth:82}, but here a misreport reshapes the priorities themselves. The difficulty is not specific to $\da$, because Rawlsian equity and strategy-proofness are incompatible in general.

\begin{proposition}[Incompatibility with strategy-proofness]\label{prop:impossible}
There exist an assignment problem and a tie-breaking profile for which no Rawlsian equitable rule is strategy-proof.
\end{proposition}

\begin{proof} Let $N=\{1,2,3\}$ and $O=\{a,b,c\}$, and let $\pi$ be the tie-breaking profile with $\pi_o:1\succ2\succ3$ for each $o\in O$.

Consider first the profile $P^{0}$ at which every agent has the ordering $a\succ b\succ c$. The unique Rawlsian equitable allocation at $P^{0}$ is $(\ap{1}{a},\ap{2}{b},\ap{3}{c})$: agent $1$ must receive $a$, since otherwise she Rawlsian envies its holder, with whom she has equal claims to $a$ and whom she precedes under $\pi_a$; and once agent $1$ holds $a$, agent $2$ must receive $b$ for the same reason.

Now let $P^{1}$ be the profile
\[ P^{1}_1, P^{1}_3:\ a\succ b\succ c, \qquad P^{1}_2:\ b\succ a\succ c. \]

The unique Rawlsian equitable allocation at $P^{1}$ is $(\ap{1}{b},\ap{2}{a},\ap{3}{c})$. Each of the other five allocations has a Rawlsian envy pair: at $(\ap{1}{a},\ap{2}{b},\ap{3}{c})$, agent $3$ Rawlsian envies agent $2$ over $b$; at $(\ap{1}{a},\ap{2}{c},\ap{3}{b})$, agent $2$ Rawlsian envies agent $1$ over $a$; at $(\ap{1}{b},\ap{2}{c},\ap{3}{a})$, agent $2$ Rawlsian envies agent $3$ over $a$; at $(\ap{1}{c},\ap{2}{a},\ap{3}{b})$, agent $1$ Rawlsian envies agent $3$ over $b$, since the two have equal claims to $b$ and agent $1$ precedes agent $3$ under $\pi_b$; and at $(\ap{1}{c},\ap{2}{b},\ap{3}{a})$, agent $1$ Rawlsian envies agent $2$ over $b$. By \autoref{thm:existence}, some allocation is Rawlsian equitable at $P^{1}$, so $(\ap{1}{b},\ap{2}{a},\ap{3}{c})$ is that allocation.

Therefore, every Rawlsian equitable rule assigns agent $2$ the object $b$ at $P^{0}$ and the object $a$ at $P^{1}$. Suppose the true profile is $P^{0}$. If agent $2$ misreports her preferences as
\[P'_2:\ b\succ a\succ c,\]
she changes the reported profile from $P^{0}$ to $P^{1}$, and her assignment changes from $b$ to $a$. Since agent $2$ prefers $a$ to $b$, she is better off after the misreport, contradicting strategy-proofness. Therefore, there is no Rawlsian equitable rule that is strategy-proof.
\end{proof}

In particular, since $\da$ is a Rawlsian equitable rule by \autoref{thm:agent-optimal}, \autoref{prop:impossible} yields the following.
\begin{corollary}[Manipulability of DA]\label{cor:da-not-sp}
There exist an assignment problem and a tie-breaking profile for which $\da$ is not strategy-proof.
\end{corollary}

\subsection{Efficiency}\label{sec:eada}

Rawlsian equity also conflicts with efficiency.

\begin{proposition}[Pareto inefficiency of DA]\label{prop:da-notefficient}
There exist an assignment problem, a preference profile, and a tie-breaking profile at which the agent-optimal Rawlsian equitable allocation is not Pareto efficient.
\end{proposition}

\begin{proof}
Let $N=\{1,2,3\}$ and $O=\{a,b,c\}$, with preferences
\[ P_1, P_2:\ a\succ b\succ c, \qquad P_3:\ b\succ a\succ c, \]
and let $\pi$ be the tie-breaking profile with $\pi_o:1\succ2\succ3$ for each $o\in O$. The induced priorities are
\[ \pi^{P}_a:3\succ1\succ2,\qquad \pi^{P}_b:1\succ2\succ3,\qquad \pi^{P}_c:1\succ2\succ3.\]

Agent-proposing deferred acceptance yields $\da(P)=(\ap{1}{b}, \ap{2}{c}, \ap{3}{a})$. By \autoref{thm:agent-optimal}, this is the agent-optimal Rawlsian equitable allocation. It is not Pareto efficient: the allocation $(\ap{1}{a}, \ap{2}{c}, \ap{3}{b})$ Pareto dominates it, strictly benefiting agents $1$ and $3$ while leaving agent $2$'s assignment unchanged.
\end{proof}

The conflict is not with $\da$ but with Rawlsian equity itself.

\begin{corollary}[Incompatibility with efficiency]\label{cor:no-eff-equitable}
There exist an assignment problem, a preference profile, and a tie-breaking profile at which no allocation is both Rawlsian equitable and Pareto efficient.
\end{corollary}

\begin{proof}
Take the agents, objects, preference profile $P$, and tie-breaking profile constructed in the proof of \autoref{prop:da-notefficient}, and suppose some allocation $\mu$ is both Rawlsian equitable and Pareto efficient. By \autoref{thm:agent-optimal}, $\da(P)\succeq_N\mu$. If some agent strictly preferred $\da(P)$ to $\mu$, then $\da(P)$ would Pareto dominate $\mu$, contradicting the Pareto efficiency of $\mu$; hence $\da(P)=\mu$. But then $\da(P)$ is Pareto efficient, contradicting \autoref{prop:da-notefficient}.
\end{proof}

The reduction that provides existence and agent-optimality also resolves the conflict: the \emph{efficiency-adjusted deferred acceptance} rule of \citet{kesten:10a}, run on the constructed priorities, selects an efficient allocation that satisfies \emph{Rawlsian neutrality}, a weaker equity requirement that we define next.

\begin{definition}[Rawlsian neutrality and efficiency] At a profile $P\in\mathcal{P}^n$, an allocation $\mu\in\mathcal{M}$ is \textbf{Rawlsian neutral} if the following holds for every allocation $\nu\in\mathcal{M}$: if $\nu(i)\mathrel{P_i}\mu(i)$ for some agent $i\in N$ who Rawlsian envies another agent under $\mu$, then $\mu(j)\mathrel{P_j}\nu(j)$ for some agent $j\in N$ who Rawlsian envies another agent under $\nu$. An allocation is \textbf{Rawlsian efficient} if it is Rawlsian neutral and Pareto efficient. \end{definition}

Every allocation that is Rawlsian equitable is Rawlsian neutral since it has no Rawlsian envy pair at all.

We now define the rule that will select the unique Rawlsian efficient allocation. Call an object \textbf{underdemanded} at an allocation if no agent assigned under that allocation prefers it to her own assignment. At a deferred-acceptance outcome an underdemanded object always exists: any object that receives a proposal in the terminal round is one. Such an object was empty until that round, because a proposal to an occupied object forces a rejection and an occupied object never becomes empty. It therefore rejects no agent during the run, and since an agent prefers an object to her assignment only if that object rejected her, no agent prefers this object to her own.

The \textbf{efficiency-adjusted deferred acceptance} rule $\eada$ computes its outcome $\eada(P)$ in rounds. \begin{itemize} \item \emph{Round $1$.} Run agent-proposing deferred acceptance with priorities $\pi^{P}$ on all of $N$ and $O$, producing $\da(P)$. Permanently assign each underdemanded object to the agent holding it, and remove those agents and objects from the market. \item \emph{Round $k>1$.} Run agent-proposing deferred acceptance on the market left by round $k-1$, with $\pi^{P}$ restricted to the agents and objects that remain. Permanently assign each object that is underdemanded at this outcome to its holder, and remove those agents and objects. \end{itemize} Since an underdemanded object exists in every round, at least one agent is settled per round, so the procedure ends after at most $n$ rounds; $\eada(P)$ is the allocation formed by the permanent assignments made across all rounds. In every round, the priorities are $\pi^{P}$ restricted to the remaining agents and objects, not priorities recomputed from ranks within the reduced market. The rule is the full-consent case of \citet{kesten:10a}, in the streamlined form of \citet{tang/yu:14}.

\begin{theorem}[Uniqueness, EADA]\label{thm:eada} For each $P\in\mathcal{P}^n$, there is a unique Rawlsian efficient allocation at $P$; it is $\eada(P)$, and either $\eada(P)=\da(P)$ or $\eada(P)$ Pareto dominates every allocation that is Rawlsian equitable. \end{theorem}

To illustrate, consider the profile constructed in the proof of \autoref{prop:da-notefficient}, where $\da(P)=(\ap{1}{b},\ap{2}{c},\ap{3}{a})$. Object $c$ is underdemanded, since no agent prefers it to her assignment. Fixing agent $2$ at $c$ and rerunning deferred acceptance on the remaining agents and objects yields $(\ap{1}{a},\ap{2}{c},\ap{3}{b})$, so $\eada(P)=(\ap{1}{a},\ap{2}{c},\ap{3}{b})$, which is exactly the Pareto improvement used in the proposition.

Rawlsian efficiency admits an equivalent description that dispenses with the neutrality clause: an allocation is Rawlsian efficient if and only if no allocation makes any agent better off unless it makes worse off some agent who Rawlsian envies another under the new allocation \citep{reny:22}. In the terminology of \citet{ehlers/morrill:20}, the legal set for this market is its von Neumann--Morgenstern stable set, whose agent-optimal element coincides with the unique Rawlsian efficient allocation; hence $\eada(P)$ is also the agent-optimal legal allocation.

The model assumes that objects have unit capacity and that agents and objects are equally numerous, but both assumptions can be relaxed as long as every agent is assigned an object. The construction of $\pi^{P}$ uses only ranks and the tie-breaking orders, so it applies unchanged when objects have capacities and aggregate capacity suffices to assign every agent; the auxiliary market is then one of college admissions with strict priorities, where deferred acceptance, the lattice structure of the stable set, and the characterization of the efficiency-adjusted rule all remain available. The impossibilities extend as well, since a market with capacities contains our examples as the special case in which every capacity is one. If agents may instead remain unassigned, outside options become necessary; that is a genuine extension rather than a relabeling, since ranks would then have to be taken among acceptable objects, and we do not pursue it here.

\subsection{Characterization}
\label{sec:characterization}

We now characterize Rawlsian equity as a fairness notion rather than as a selection rule. The primitive is the adjudication of an objection raised by one agent against another over a contested object. Throughout, we restrict attention to notions that adjudicate objections \emph{pairwise} and on the basis of \emph{ranks} alone: whether an objection is upheld may depend on the two agents, the contested object, and the positions the two agents assign to it, but on nothing else. These two restrictions are maintained assumptions, not axioms; they fix the scope of the analysis. Criteria that compare complete allocations, such as the leximin criteria discussed above, do not violate any condition below; they lie outside the class we consider. Within this class, four conditions characterize Rawlsian equity.

For a tie-breaking profile $\pi\in\Pi^n$, let $\re_\pi(P)$ denote the set of allocations that are Rawlsian equitable at $\pi$ when the profile is $P$.

\begin{definition}[Rank-local objection system]
\label{def:rank-local-objections}
A \textbf{rank-local objection system} is a collection
\[
\kappa=\left(\kappa_o^{ij}\right)_{\substack{o\in O\\ i,j\in N,\ i\neq j}},
\qquad
\kappa_o^{ij}:\{1,\ldots,n\}^2\longrightarrow\{0,1\}.
\]
We read $\kappa_o^{ij}(k,\ell)=1$ as: agent $i$'s objection against agent $j$ over object $o$ is upheld when $i$ ranks $o$ in position $k$ and $j$ ranks it in position $\ell$. The fairness correspondence generated by $\kappa$ is
\[
F^\kappa(P)
=
\left\{
\mu\in\mathcal{M}:
\begin{array}{l}
\text{there are no distinct agents $i,j\in N$ such that}\\
\mu(j)\mathrel{P_i}\mu(i)
\text{ and }
\kappa_{\mu(j)}^{ij}\bigl(\rank(P_i,\mu(j)),\rank(P_j,\mu(j))\bigr)=1
\end{array}
\right\}.
\]
\end{definition}

An allocation is unfair, then, exactly when some envious agent has an upheld objection against the holder of the envied object. In particular, every envy-free allocation belongs to $F^\kappa(P)$, whatever $\kappa$ may be.

\begin{definition}[Conditions on objection systems]
\label{def:objection-axioms}
A rank-local objection system $\kappa$ is:
\begin{enumerate}
\item \textbf{decisive} if $\kappa_o^{ij}(k,\ell)+\kappa_o^{ji}(\ell,k)=1$ for every object $o$, distinct agents $i,j$, and ranks $k,\ell$;
\item \textbf{compensatory} if $\kappa_o^{ij}(k,\ell)=1$ for every object $o$, distinct agents $i,j$, and ranks $k>\ell$;
\item \textbf{tie-invariant} if $\kappa_o^{ij}(k,k)=\kappa_o^{ij}(\ell,\ell)$ for every object $o$, distinct agents $i,j$, and ranks $k,\ell$;
\item \textbf{attainable} if $F^\kappa(P)\neq\emptyset$ for every $P\in\mathcal{P}^n$.
\end{enumerate}
\end{definition}

Compensation is the normative content: between two agents who both want an object, the one for whom it is the less favorable fallback has the stronger claim. Reversing the inequality yields the fitness principle instead, and we return to that alternative in the Appendix. Tie-invariance requires that the resolution of a tie not depend on the common position at which the tie occurs; two agents who both rank an object second are adjudicated as two agents who both rank it fifth would be. Attainability requires that some allocation be deemed fair at every preference profile; a criterion that can condemn every allocation offers no guidance.

Decisiveness is the most restrictive of the four conditions. It requires every pairwise comparison of claims to have exactly one winner, and so rules out leaving equal claims unadjudicated. The requirement is substantive: the system that upholds no objection under equal claims is attainable and generates the criterion that \autoref{prop:tiebreak} characterizes as Rawlsian equity at some tie-breaking profile. Decisiveness reflects the view that a fairness notion which recognizes a conflict should settle it, and that declining to adjudicate equal claims leaves the notion silent exactly where two agents are situated identically.

\begin{theorem}[Characterization of Rawlsian equity]
\label{thm:rawlsian-characterization}
A rank-local objection system $\kappa$ is decisive, compensatory, tie-invariant, and attainable if and only if there is a tie-breaking profile $\pi\in\Pi^n$ such that, for every object $o\in O$, distinct agents $i,j\in N$, and ranks $k,\ell\in\{1,\ldots,n\}$,
\[
\kappa_o^{ij}(k,\ell)=1
\quad\Longleftrightarrow\quad
k>\ell
\quad\text{or}\quad
\bigl[k=\ell\text{ and }\pi_o(i)<\pi_o(j)\bigr].
\]
In that case $F^\kappa(P)=\re_\pi(P)$ for every $P\in\mathcal{P}^n$.
\end{theorem}

The proof, in the Appendix, relies on a lemma showing that attainability forces the resolution of equal claims to be transitive: a cyclic tie-break leaves some preference profiles with no fair allocation.

The conditions are imposed on the objection system rather than on the correspondence it generates, and the two are not equivalent in general. Some rank pairs are never triggered: an agent who ranks $o$ last cannot prefer $o$ to her own object, so $\kappa_o^{ij}(n,n)$ has no bearing on $F^\kappa$. What \autoref{thm:rawlsian-characterization} pins down uniquely is the adjudication rule; the correspondence it generates is determined by every configuration that can arise.

The characterization and \autoref{prop:tiebreak} answer different questions. \autoref{prop:tiebreak} fixes a preference profile and an allocation and asks whether some tie-breaking profile renders that allocation Rawlsian equitable. \autoref{thm:rawlsian-characterization} instead imposes conditions on a fairness notion across all preference profiles and recovers a single tie-breaking profile that represents the notion throughout.

\section{Conclusion}

We have introduced Rawlsian equity, a fairness notion for the assignment problem that grades competing claims to an object by the positions the claimants assign to it. Rawlsian equitable allocations are exactly the stable allocations of an auxiliary market whose priorities are built from the reported preferences. They therefore always exist, form a lattice, and contain an agent-optimal element, which the agent-proposing deferred acceptance algorithm computes. Equity carries two costs: no Rawlsian equitable rule is strategy-proof, and no allocation is, in general, both Rawlsian equitable and Pareto efficient. Weakening equity to Rawlsian neutrality resolves the second cost, since a unique Rawlsian efficient allocation exists and the efficiency-adjusted deferred acceptance rule selects it. Finally, four conditions on the adjudication of pairwise objections characterize Rawlsian equity, and the characterization recovers the tie-breaking profile from the requirement that fair allocations exist.

Rawlsian equity and strategy-proofness are incompatible because reports shape the priorities agents face. This leaves open how manipulability is structured in our setting, and whether some Rawlsian equitable rule is less exposed to it than others. One natural benchmark is the weaker requirement of not being obviously manipulable \citep{troyan/morrill:20}; whether any Rawlsian equitable rule meets it we leave for future work. The question is not idle: \citet{troyan/delacretaz/kloosterman:20} show that no essentially stable mechanism is obviously manipulable. Since their relaxation, unlike stability itself, is compatible with Pareto efficiency, weakening a fairness requirement can restore efficiency without exposing the rule to obvious manipulation.

\bibliography{ref}

\clearpage

\appendix

\section*{Appendix: Omitted proofs}

\subsection*{Proof of \autoref{prop:tiebreak}}
\begin{proof}
We first prove the equivalence for every tie-breaking profile. Let $\mu$ be absolutely Rawlsian equitable, and consider agents $i$ and $j$ with $\mu(j)\mathrel{P_i}\mu(i)$. Since $\mu$ admits no weak Rawlsian envy, agent $i$'s claim to $\mu(j)$ is neither stronger than nor equal to agent $j$'s. Rawlsian envy requires one of these two, so $i$ does not Rawlsian envy $j$ at any tie-breaking profile. Hence, $\mu$ is Rawlsian equitable at every tie-breaking profile.

Conversely, let $\mu$ be Rawlsian equitable at every tie-breaking profile, and suppose towards a contradiction that agent $i$ weakly Rawlsian envies agent $j$, so that $\mu(j)\mathrel{P_i}\mu(i)$ and agent $i$'s claim to $\mu(j)$ is stronger than or equal to agent $j$'s. If the claim is strictly stronger, then $i$ Rawlsian envies $j$ at every tie-breaking profile, since the tie-breaking profile enters the definition only under equal claims. If the claims are equal, choose a tie-breaking profile whose tie-breaking order for $\mu(j)$ ranks $i$ ahead of $j$; then $i$ Rawlsian envies $j$ at that profile. Either case contradicts the assumption. Thus, under $\mu$, no agent weakly Rawlsian envies another, so $\mu$ is absolutely Rawlsian equitable.

We turn to the equivalence for some tie-breaking profile. Suppose first that some agent $i$ prefers $\mu(j)$ to $\mu(i)$ while having a stronger claim to $\mu(j)$ than agent $j$. As noted above, $i$ then Rawlsian envies $j$ at every tie-breaking profile, so $\mu$ is Rawlsian equitable at none.

Conversely, suppose no agent prefers another agent's object to her own while having a stronger claim to it. Let $\pi$ be the tie-breaking profile in which, for each object $o$, the tie-breaking order $\pi_o$ ranks the agent $\mu^{-1}(o)$ first and the remaining agents in any fixed order. Consider agents $i$ and $j$ with $\mu(j)\mathrel{P_i}\mu(i)$. By hypothesis, agent $i$'s claim to $\mu(j)$ is not stronger than agent $j$'s. If the two claims are equal, then, because $\pi_{\mu(j)}$ ranks $j=\mu^{-1}(\mu(j))$ first, $\pi_{\mu(j)}(i)>\pi_{\mu(j)}(j)$. In neither case does $i$ Rawlsian envy $j$ at $\pi$. Hence, $\mu$ is Rawlsian equitable at $\pi$.
\end{proof}

\subsection*{Proof of \autoref{thm:eada}}
\begin{proof}
By the equivalence in the proof of \autoref{thm:agent-optimal}, the allocations that are Rawlsian equitable at $P$ are exactly those stable with respect to $\pi^{P}$, and an agent Rawlsian envies another under an allocation if and only if that agent's priority is violated in the market with priorities $\pi^{P}$. Rawlsian neutrality and Rawlsian efficiency are therefore, respectively, the priority-neutrality and priority-efficiency of \citet{reny:22} for that market.

Since $\pi^{P}\in\Pi^n$ and each object has a single unit, \citet{reny:22} applies: a unique Rawlsian efficient allocation exists, it is the full-consent efficiency-adjusted deferred acceptance outcome $\eada(P)$ of \citet{kesten:10a}, and every agent weakly prefers it to every Rawlsian neutral allocation. Every allocation that is Rawlsian equitable is stable with respect to $\pi^{P}$, hence Rawlsian neutral, so every agent weakly prefers $\eada(P)$ to it.

If $\eada(P)$ equals some Rawlsian equitable allocation, then, being weakly preferred by every agent to the agent-optimal allocation $\da(P)$, it equals $\da(P)$. Otherwise, $\eada(P)$ differs from every Rawlsian equitable allocation while every agent weakly prefers it to each, so $\eada(P)$ Pareto dominates every allocation that is Rawlsian equitable.
\end{proof}

\subsection*{Proof of \autoref{thm:rawlsian-characterization}}
The proof relies on the following lemma, which is the heart of the characterization. It says that the acyclicity of the tie-breaking is not an assumption we impose but a consequence of requiring that fair allocations exist.

\begin{lemma}[Attainability forces acyclicity]
\label{lem:acyclicity}
Let $\kappa$ be decisive, compensatory, tie-invariant, and attainable. For each object $o\in O$, define a binary relation $T_o$ on $N$ by
\[
i\mathrel{T_o}j
\quad\Longleftrightarrow\quad
\kappa_o^{ij}(k,k)=1
\text{ for some }k\in\{1,\ldots,n\}.
\]
Then $T_o$ is a strict linear order on $N$.
\end{lemma}

\begin{proof}
Tie-invariance implies that the defining condition is independent of the choice of $k$, so $T_o$ is well defined. Decisiveness implies that $T_o$ is a tournament: for distinct $i$ and $j$, exactly one of $i\mathrel{T_o}j$ and $j\mathrel{T_o}i$ holds.

We claim that $T_o$ is transitive. Suppose it is not. Every non-transitive tournament contains a directed three-cycle, so there are distinct agents, which we label $1$, $2$, and $3$, with
\[
1\mathrel{T_o}2,\qquad 2\mathrel{T_o}3,\qquad 3\mathrel{T_o}1.
\]

Suppose first that $n=3$. Consider a profile at which all three agents rank $o$ first. Under any allocation, the holder of $o$ is preceded under $T_o$ by one of the other two agents. That agent prefers $o$ to her assignment and has an equal claim to it, so her objection is upheld. Hence $F^\kappa(P)=\emptyset$, contradicting attainability.

Suppose now that $n>3$, and fix an object $q\neq o$. Consider the profile at which agents $1$, $2$, and $3$ rank $q$ first and $o$ second, every other agent ranks $o$ first and $q$ second, and the remaining objects are ordered arbitrarily below. Let $\mu$ be an allocation.

If some agent $h\notin\{1,2,3\}$ receives $o$, then, because at most one agent receives $q$, some $i\in\{1,2,3\}$ receives neither $q$ nor $o$. Agent $i$ prefers $o$ to her assignment, ranks $o$ second, and $h$ ranks $o$ first, so compensation upholds $i$'s objection against $h$.

Otherwise, some $j\in\{1,2,3\}$ receives $o$. Let $i\in\{1,2,3\}$ be the agent with $i\mathrel{T_o}j$. If $\mu(i)\neq q$, then agent $i$ prefers $o$ to her assignment and has an equal claim to it, since both $i$ and $j$ rank $o$ second, so her objection against $j$ is upheld. If instead $\mu(i)=q$, then both $o$ and $q$ are held by agents in $\{1,2,3\}$, so any agent $h\notin\{1,2,3\}$ receives neither. Agent $h$ prefers $q$ to her assignment, ranks $q$ second, and $i$ ranks $q$ first, so compensation upholds $h$'s objection against $i$.

Every allocation, therefore, admits an upheld objection, so $F^\kappa(P)=\emptyset$, again contradicting attainability. Hence, $T_o$ is transitive, and being complete, asymmetric, and transitive, it is a strict linear order.
\end{proof}

With \autoref{lem:acyclicity} in hand, we prove the theorem.

\begin{proof}
Suppose there is a tie-breaking profile $\pi$ satisfying the displayed equivalence. Compensation and tie-invariance are immediate. For decisiveness, if $k>\ell$ then $\kappa_o^{ij}(k,\ell)=1$ and $\kappa_o^{ji}(\ell,k)=0$, while if $k=\ell$ then exactly one of $\pi_o(i)<\pi_o(j)$ and $\pi_o(j)<\pi_o(i)$ holds. The condition defining $F^\kappa$ now coincides with the definition of Rawlsian envy at $\pi$, so $F^\kappa(P)=\re_\pi(P)$ for every $P$, and attainability follows from \autoref{thm:existence}.

Conversely, let $\kappa$ be decisive, compensatory, tie-invariant, and attainable. Compensation gives $\kappa_o^{ij}(k,\ell)=1$ whenever $k>\ell$. If $k<\ell$, then compensation applied to the reverse comparison gives $\kappa_o^{ji}(\ell,k)=1$, so decisiveness gives $\kappa_o^{ij}(k,\ell)=0$. For equal claims, \autoref{lem:acyclicity} shows that $T_o$ is a strict linear order; let $\pi_o$ be the tie-breaking order with $\pi_o(i)<\pi_o(j)$ if and only if $i\mathrel{T_o}j$, and let $\pi=(\pi_o)_{o\in O}$. By tie-invariance, $\kappa_o^{ij}(k,k)=1$ if and only if $\pi_o(i)<\pi_o(j)$, which is the displayed equivalence.
\end{proof}

\subsection*{Independence of the conditions}
Each condition is indispensable.

\begin{enumerate}
\item \emph{Decisiveness.} Let $\kappa_o^{ij}(k,\ell)=1$ if and only if $k>\ell$, so that equal claims never sustain an objection. This system is compensatory and tie-invariant, and by \autoref{prop:tiebreak} it generates
\[
F^\kappa(P)=\bigcup_{\pi\in\Pi^n}\re_\pi(P),
\]
which contains $\re_\pi(P)$ for any fixed $\pi$ and is therefore nonempty by \autoref{thm:existence}. It is not decisive, and it is not $\re_\pi$ for any $\pi$: at a profile where all agents share one preference ordering, every allocation belongs to $F^\kappa(P)$, whereas $\re_\pi(P)$ contains only allocations that give the commonly top-ranked object to the agent whom $\pi$ ranks first for it.

\item \emph{Compensation.} Fix $\pi\in\Pi^n$ and let $\kappa_o^{ij}(k,\ell)=1$ if and only if $k<\ell$, or $k=\ell$ and $\pi_o(i)<\pi_o(j)$. This system implements the fitness principle, awarding the object to the agent who ranks it higher. It is decisive and tie-invariant, and attainable by the argument of \autoref{thm:existence} applied to the priorities that order agents by increasing rank of the object, with ties broken by $\pi_o$.

\item \emph{Tie-invariance.} For each object $o$ and rank $k$, fix a strict linear order $\pi_{o,k}$ on $N$, and resolve a tie at rank $k$ by $\pi_{o,k}$. This system is decisive and compensatory, and attainable by the same argument, since each object still induces a strict priority order at every profile. Whenever $\pi_{o,k}$ varies with $k$, it is not tie-invariant and hence not of the form in \autoref{thm:rawlsian-characterization}.

\item \emph{Attainability.} Suppose $n\geq 3$ and resolve equal claims by a tournament containing a directed three-cycle, the same tournament at every rank. This system is decisive, compensatory, and tie-invariant, and by \autoref{lem:acyclicity} it is not attainable.
\end{enumerate}

\end{document}